\documentclass[11pt]{article}
\pdfoutput=1
\usepackage{jheppub}

\usepackage{physics}
\usepackage{mathtools}
\usepackage{amsmath, amssymb, amsthm, amsbsy, bbm, braket}
\usepackage{dsfont}
\usepackage[percent]{overpic}
\usepackage{xcolor}
\usepackage[utf8]{inputenc}
\usepackage{enumitem}
\usepackage[normalem]{ulem}
\usepackage{etoolbox}
\usepackage{cancel}
\usepackage{xargs}                      % Use more than one optional parameter in a new command

\newtheorem{thm}{Theorem}[section]

\theoremstyle{remark}

\theoremstyle{definition}

\newtheorem{claim}[thm]{Claim}

\newcommand{\calC}{\mathcal{C}}

\newcommand{\bea}{\begin{eqnarray}}
\newcommand{\eea}{\end{eqnarray}}
\newcommand{\be}{\begin{equation}}
\newcommand{\ee}{\end{equation}}
\newcommand{\ba}{\begin{align}}
\newcommand{\ea}{\end{align}}

\newcommand\rref[1]{(\ref{#1})}

\newlength{\slength}
\newcommand{\vn}{\vec{n}}

\newcommandx{\jamie}[2][1=]{\todo[author=Jamie,linecolor=red,backgroundcolor=red!25,bordercolor=red,#1]{#2}}

\usepackage{subfiles}
\title{The Spectrum of Boundary States in Symmetric Orbifolds}
\author{Alexandre Belin,$^a$ Shovon Biswas,$^b$ James Sully$^b$}

\affiliation[\,a]{CERN, Theory Division, 1 Esplanade des Particules, Geneva 23, CH-1211, Switzerland}
\affiliation[\,b]{Department of Physics and Astronomy, University of British Columbia,\\
6224 Agricultural Road, Vancouver, B.C.\ V6T 1Z1, Canada.}

\emailAdd{a.belin@cern.ch, shovon432@gmail.com, jamie.sully@gmail.com}

\abstract{We give an explicit construction of the complete set of Cardy boundary states that respect the extended chiral algebra for symmetric product orbifolds. The states are labelled by a choice of seed theory boundary states as well as a choice of representations of the symmetric group. At large $N$, we analyze the BCFT data which is relevant for holography, namely the boundary entropy and the one-point functions of single-trace operators. In some cases, typical boundary states are compatible with a bulk description in terms of an End-of-the-World brane along with backreacted matter fields. We discuss the significance of these results for the AdS/BCFT correspondence.}

\begin{document}

\begin{flushright}
\hfill{\tt CERN-TH-2021-149}
\end{flushright}

\maketitle

%%%%%%%%%%%%%%%%%%%%%%%%%%%%%%%%%%%%%%%%%%%%
\section{Introduction}\label{sec:intro}
%%%%%%%%%%%%%%%%%%%%%%%%%%%%%%%%%%%%%%%%%%%%

What does a typical conformal boundary condition look like for a holographic conformal field theory? 
Do most boundary conditions have a large number of degrees of freedom that cap of the dual geometry deep in the infrared, leading to a large bulk geometry? Or does a typical boundary condition eat up most of the bulk degrees of freedom and destroy all but a sliver of the dual bulk geometry? And are there any limits on how large or how small the dual geometry can become?

These questions are of great interest with the recent upsurge in applications of holographic boundary conformal field theory (BCFT). BCFTs have been used as examples of a gravitational system (the bulk dual of the boundary) coupled to an auxiliary bath (the ambient CFT) to study black hole evaporation and entanglement islands \cite{1908.10996,1910.12836,2006.04851,2010.00018,2012.04671,2105.00008,Hollowood:2021wkw,Geng:2021iyq,Ageev:2021ipd}.
BCFTs have likewise seen renewed interest as examples of cosmological spacetimes (the bulk dual of the boundary) embedded in a higher-dimensional spacetime (the bulk dual of the ambient CFT) \cite{1810.10601,1907.06667,2008.02259,2102.05057}. 
In both of these settings, the localization of gravity to a gravitational `End-of-the-World' (ETW) brane describing the boundary condition requires that the boundary condition have a large number of degrees of freedom and the ambient bulk geometry to be large.  
Holographic boundary conditions have also been used as examples of black hole microstates \cite{1803.04434,1810.10601,Kourkoulou:2017zaj} and there is a comparably longer history of their use in studying quantum quenches (see \cite{1603.02889} for a review).

A holographic BCFT is generally regarded in two ways. 
In bottom-up constructions, it has been proposed that a BCFT should be dual to a bulk geometry that is cut-off by a semiclassical ETW brane \cite{hep-th/0011156,hep-th/0105132,1105.5165,1108.5152}. 
However, there are few examples where this picture has been proven to be a consistent approximation (see \cite{1108.5152} for one example, though). 
In explicit top-down constructions \cite{1111.6912,1205.5303,1304.4389,1312.5477,0705.0022,0705.0024,1106.1870,0804.2902,0807.3720}, where the bulk dual is on firmer ground, the bulk geometry usually caps off smoothly when a cycle pinches off in the higher-dimensional geometry. 
These geometries are not necessarily well-approximated by a thin, semi-classical ETW brane.

The question of what a typical boundary condition looks like in a holographic CFT is difficult to answer. 
In particular, in a generic large $C_T$ CFT we don't understand how to construct consistent boundary conditions let alone understand typical properties of the space of all boundary conditions. 
To make progress, we must choose a less ambitious goal.\footnote{One can also ask the inverse question: assuming an ETW brane model, what does it imply about statistical properties of boundary states? This question is considered in \cite{Miyaji:2021ktr}.}  

One possibility is to study particular examples of supersymmetric large $C_T$ conformal field theories where we have much greater control and calculational power---namely those considered in the top-down models. 
In this direction, in \cite{Raamsdonk:2020tin,VanRaamsdonk:2021duo} the authors studied 4D $\mathcal{N}=4$ supersymmetric Yang-Mills theory coupled to a 3D superconformal field theory on its boundary (building on previous work in \cite{0804.2902,0807.3720,0705.0022,0705.0024,1106.1870,1106.4253,1206.2920,1812.00923,2005.10833}). 
In this case the supergravity dual of the BCFT is known explicitly. 
They found that the considered boundary conditions can have arbitrarily many degrees of freedom with an arbitrarily large bulk geometry, but they do not explicitly consider what is typical among the conditions considered.

Here, we take a different approach by studying two dimensional CFTs where Virasoro symmetry gives us greater analytical control. 
However, studying the boundary conditions of a generic large-$c$ CFT is still beyond our reach because the questions we would like to answer are not kinematically fixed even with the help of Virasoro symmetry. Instead we will focus on symmetric orbifold CFTs, 
\begin{equation}
    \mathcal{C}^{\otimes N}/S_N
\end{equation}
for some fixed seed CFT $\calC$, in the limit of large $N$. This CFT has a large central charge $N c_{\scriptscriptstyle \calC}$ and shares many characteristics with a holographic 2D CFT. For example, the finite temperature partition function is universal at large $N$ and follows that of AdS$_3$ gravity \cite{Keller:2011xi}, while the correlation functions obey large $N$ factorization  \cite{Pakman:2009zz}. There are of course differences with true holographic theories; for instance the growth of states in symmetric orbifolds is exponential in the energy, much faster than the growth dictated by a supergravity dual \cite{Belin:2016dcu,Belin:2018oza}. Moreover, the theories are not chaotic and don't display exponential growth in out-of-time-ordered correlation functions \cite{Belin:2017jli}, which can be expected since symmetric orbifolds are free discrete gauge theories. This vast landscape of large $c$ CFTs is thus the perfect place to start a classification, as they give us a lot of control while mimicking as best they can the properties of a true holographic CFT.\footnote{There are also other orbifold theories known as permutation orbifolds where one quotients by a subgroup of $S_N$ that display similar properties, at least for some choices of subgroups \cite{Haehl:2014yla,Belin:2014fna,Belin:2015hwa,Keller:2017rtk}. We will not discuss their boundary states in this paper, but it would be a straightforward generalization. }

Even when the seed theory has finitely many conformal boundary conditions (as for minimal models), we expect the orbifold theory to have infinitely many boundary conditions corresponding to the infinitely many Virasoro primary operators it contains. 
In this paper, we will characterize boundary conditions in the orbifold theory that also respect the extended symmetry algebra of the theory. 
That is, we look for boundary states annihilated not just by
\begin{equation}
    T(z) - \overline{T}(\overline{z})
\end{equation}
but also by
\begin{equation} \label{chiralalgebracond}
    W^i(z) - \Omega \overline{W}^i(\overline{z})
\end{equation}
for generators $W^i,\overline{W}^i$ of the extended symmetry algebra, for some choice of $\Omega$ an automorphism of the algebra that fixes the stress tensor of the orbifold theory.\footnote{In this paper, we will consider only the trivial automorphism $\Omega=1$ but we comment on more general $\Omega$ in the discussion.}  
In particular, the boundary states we construct in this paper will all be built out of products of the Ishibashi states of the seed theory and corresponding Ishibashi states in the twisted sectors.

We are able to characterize the complete set of these boundary states. The boundary conditions are classified by a vector of multiplicities $\vn$ labelling how many times each seed boundary condition appears in the state, as well as a vector of permutations (really conjugacy classes) $\vec{r}$ labelling the contribution of each `twisting' of the seed states. We find that the consistent boundary states are
\begin{align}
   | \vec{n},\vec{r} \rangle = \frac{1}{\sqrt{N!}}\sum_{h \in S_N}\prod_{i=1}^{n_b} \left(  \sum_{g_i \in S_{n_i}(N_i) } \frac{\chi^{r_i}(g_i)}{n_i!} |(a_i)_{h g_i h^{-1}} \rangle \right) \,,
   \label{eq:general-state-intro}
\end{align}
where the $\chi^{r_i}$ are characters of the symmetric group and where the states $\ket{(a_i)_g}$ are Ishibashi states in the twist $g$ sector. A more precise definition of each component is given in section 3.1 and below equation \rref{eq:general-state}.

We prove that these states satisfy Cardy's consistency conditions and are thus good boundary states. This entails showing that the partition function in the open string channel has an expansion in terms of integer coefficients, which we prove. We also explain that these coefficients have a nice interpretation in terms of counting representations of subgroups of the symmetric group a certain number of times, determined by the branching rules of the characters $\chi^{r_i}$.

Given the states \rref{eq:general-state-intro}, we proceed to analyze the BCFT data relevant for holography: the boundary entropy $g$ and the one-point functions of single-trace operators $p_O$ both in the untwisted and twisted sectors. We are most interested in typical boundary states for a fixed seed CFT, and typicality ends up being highly sensitive to whether the seed theory has a finite or infinite number of boundary states $n_b$. For \textit{typical} boundary states, we find the following results:

  \begin{align}
  n_b \textrm{ is finite}& \quad
    \begin{cases}
      g &\sim N \\
      p_{\textrm{untw}} &\sim \sqrt{N} \\
      p_{\textrm{tw}} &\sim N^0
    \end{cases} \\
    \nonumber\\
        n_b \textrm{ is infinite}&  \quad
    \begin{cases}
      g &\sim N \log N \\
      p_{\textrm{untw}} &\sim \sqrt{N} \\
      p_{\textrm{tw}} &= 0\,.
    \end{cases}
  \end{align}
From this data, we can conclude that typical boundary states in symmetric orbifolds of infinite (e.g. irrational) seed CFTs do not appear to have a nice bulk dual, since the tension of the ETW brane (or depth of the bulk geometry) is super-Planckian. For finite seed theories (e.g. a minimal model), the boundary entropies and one-point functions are consistent with having a macroscopic bulk dual.
However, it is a much stronger statement to conclude that a semiclassical bulk dual actually exists. 
For example, recent work \cite{Reeves:2021sab} has shown that the existence of a bulk causal structure puts severe constraints on the spectrum of boundary operators. This data could in principle be computed from our methods but we leave it for future work.

The outline of this paper is as follows. In Section \ref{sec:review} we give a brief review of necessary facts about boundary conformal field theories and symmetric orbifold theories. 
In Section \ref{sec:bdrystates} we derive a complete set of boundary states in our symmetry class for the symmetric orbifolds.
In Section \ref{sec:branetension} we calculate the distribution of boundary entropies and single-trace one-point functions in the large-$N$ limit.
Finally, in Section \ref{sec:discussion}, we end with some concluding remarks and open questions.

\textbf{Note added:} While finishing this work, we became aware of related work being done in \cite{gaberdiel_upcoming}. We have coordinated the release of our papers.

%%%%%%%%%%%%%%%%%%%%%%%%%%%%%%%%%%%%%%%%%%%%
\section{Review of BCFT and Symmetric Orbifolds}\label{sec:review}

In this section, we review some aspects  of two-dimensional conformal field theory that will be relevant for this paper. We start by a review of conformal boundary conditions, introducing Ishibashi states and Cardy's consistency conditions. We then turn to symmetric orbifolds and give a brief review of their construction.
%%%%%%%%%%%%%%%%%%%%%%%%%%%%%%%%%%%%%%%%%%%%
%-------------------------------------------
\subsection{Review of BCFT}\label{sec:BCFTrev}

Here, in reviewing some basic aspects of boundary conformal field theory in $d=2$, we will mostly follow  \cite{Cardy:2004hm}. 

To define a conformal field theory on a space with a boundary, we must choose a set of consistent boundary conditions. 
It's particularly useful to choose boundary conditions that leave a maximal subgroup of the conformal group unbroken, the so-called \textit{conformal boundary conditions}. 
Conformal boundary conditions impose that the off-diagonal components of the stress-tensor for the directions parallel and perpendicular to the boundary vanish. In $d=2$, this condition can be written in the form
\be \label{confbc}
T(z)=\bar{T}(\bar{z})
\ee
for $z$ on the boundary along the real axis (we take the CFT to be defined on the upper half plane). 

Using the standard logarithmic transformation, one can map the upper half-plane to an interval with a boundary on either end (and an infinitely long Euclidean time direction). 
When we impose different boundary conditions on the positive and negative real axis, we get different boundary conditions on either end of the interval. 
If we quotient the radial direction as well, our theory then lives on a finite cylinder with a boundary at either end.   
If we choose to treat the longitudinal direction of our cylinder as Euclidean time, we can interpret this as the matrix element
\be \label{transitionamplitude}
\bra{a}e^{-\beta H}\ket{b} \,,
\ee
between \textit{boundary states} $a$ and $b$ living in the original CFT Hilbert space\footnote{Technically, the boundary states themselves are non-normalizable.}. Here, $H$ is the standard CFT Hamiltonian on the circle.  

When we view our boundary conditions as states on the circle, the constraint on the stress tensor \rref{confbc} can be expressed as

\be \label{Ishibashicondition}
L_{n}\ket{a}=\bar{L}_{-n}\ket{a} \,.
\ee
It is easy to see that a primary state will not satisfy this condition. However, for a scalar primary there is unique solution given by
\be
\ket{h}\rangle=\sum_{n=0}^{\infty} \sum_{j=1}^{d_h(n)} \ket{h,n,j,\bar{n},\bar{j}} \,.
\ee
Here $k$ labels the Virasoro descendant level, and $d_h(n)$ the number of states at that level. The states $\ket{h,n,j,\bar{n},\bar{j}}$ are an orthonormal basis of Virasoro descendants of the primary operator with $h=\bar{h}$. These states are known as Ishibashi states. While they solve the conformal boundary condition, they are not yet consistent boundary states.

Proper boundary states are constrained by Cardy's consistency conditions. These conditions enforce that any two boundary conditions $a,b$ give rise to a proper Hilbert space of states living on the interval stretching between them. In other words, they require the partition function of the theory with boundary conditions $a,b$ to be a sum over states with integer multiplicities. 
We call these proper boundary states \textit{Cardy states}.

In more detail, we consider a strip of width $L$  with boundary condition $a$ and $b$ on the sides. Taking the time direction $t\sim t+T$ along the length of the strip gives a cylinder of circumference $T$. The path integral on this cylinder computes a thermal partition function for a Hilbert space $\mathcal{H}_{ab}$ on an interval with boundary conditions specified at both ends
\begin{eqnarray}
\label{ann}
    Z_{ab}^{open}&&=\text{Tr}_{\mathcal {H}_{ab}} q^{L_0-\frac{c}{24}},\quad q=e^{2\pi i \tau}\nonumber\\
    &&=\sum_h N^h_{ab}\chi_h(q).
    \end{eqnarray}
Here $\tau=\frac{iT}{2L}$ is the modular parameter of the cylinder.  Since the conformal boundary condition $T(z)=\bar T(\bar z)$ on the real line $z=\bar z$ eliminates half of the Virasoro generators, the partition function (\ref{ann}) depends only on the chiral characters. Because this partition function is a trace over states, Cardy's conditions require $N^h_{ab}$ to be integers. Moreover, we will demand that the vacuum state is unique and appears only when the two boundary states are identical, so that $N^0_{ab} =\delta_{ab}$\footnote{The fact that the vacuum only appears when the boundary states are identical means that these states are orthonormal when appropriately regulated.}. The partition function (\ref{ann}) is sometimes called the `open string' partition function and we will adopt this nomenclature here.

By interchanging the role of the space and time coordinates, we can interpret (\ref{ann}) as a transition amplitude between two boundary states $\ket{a}$ and $\ket{b}$, as in \rref{transitionamplitude}. This gives the `closed string' partition function written as
\begin{eqnarray}\label{closd}
    Z_{ab}^{closed}&&=\bra{a}e^{-L H}\ket{b}\nonumber\\
    &&=\bra{a}\Tilde{q}^{\frac{1}{2}(L_0+\Bar{L}_0-\frac{c}{12})}\ket{b}\quad\tilde q=\frac{-2\pi i}{\tau},
\end{eqnarray}
where  $H=\frac{2\pi}{T}(L_0+\Bar{L}_0-\frac{c}{12})$ is the  Hamiltonian of the CFT defined on a circle of size $T$. 

A generic superposition of Ishibashi states will correspond to a closed string amplitude, but will in general not admit an interpretation as a well-defined partition function in the open string sector. To obtain a Cardy state, one must tune the linear combination. This is similar to modular invariance for rational CFTs where one appropriately adds the right characters to the partition function so as to obtain a good modular invariant partition function.

Cardy states are thus written as the linear combinations  \cite{ishibashi1989boundary}
\begin{equation} \label{superpositionIshi}
    \ket{a}=\sum_{h}\langle\bra{h}\ket{a}\ket{h}
\rangle\equiv\sum_h a_h\ket{h}
\rangle\,.
\end{equation}
Performing a modular $S$-transform and comparing (\ref{closd}) with (\ref{ann}) gives the Cardy conditions
\begin{equation}
    \sum_{k} a^*_kb_kS^h_k=N^h_{ab} \,. \label{cardy}
\end{equation}
For RCFTs a solution to the Cardy condition (\ref{cardy}) is given by the following choice of the constants
\begin{equation}
    a_h=\frac{S_{ah}}{\sqrt{S_{0h}}} \,.
\end{equation}

\subsection*{Boundary Entropy} 

A particular quantity of interest is the boundary entropy, which measures the correction to the asymptotic density of states for the BCFT. In particular, in the large temperature limit of the open string partition function, when the closed string $\beta \rightarrow \infty$, we have 
\begin{equation}
    Z^{\mathrm{open}}_{ab} = \braket{a|0} \braket{0|b} e^{-\beta c/6} + \ldots
    \label{eq:large-temp}
\end{equation}
The boundary entropy accounts for the scaling contribution of the overlap of the boundary states with the vacuum. We define
\begin{equation}
    g_a=\log\langle a | 0\rangle \,.
\end{equation}

When the bulk dual of a BCFT is given by empty AdS terminated by an ETW brane, the boundary entropy can be simply related to the tension and location of the brane \cite{Takayanagi:2011zk}. 
More generally, whenever a BCFT has a bulk dual, the boundary entropy gives a measure of the depth of the bulk geometry; it can be related to the distance to the end of the geometry in the IR. 
In particular, for an interval of length $L$ containing the boundary, the entanglement entropy takes the universal form
\begin{equation}
\label{eq:ee}
S = \frac{c}{6} \log {\frac{2 L}{\epsilon}} + \log g_b \; .
\end{equation}
If we have a good bulk dual where the Ryu-Takayanagi formula can be used to compute this entanglement entropy, then we see that the boundary entropy determines the length of a corresponding bulk geodesic that terminates at the IR end of the bulk geometry. 
This argument extends easily to the case that the bulk dual has extra compact dimensions that pinch off in the IR. Instead of a geodesic length, we are simply measuring the volume of a surface that wraps the extra compact directions using higher-dimensional Planck units. It remains true that the boundary entropy should give us an effective notion of the depth of the bulk dual. 

Understanding the spectrum of boundary states in symmetric orbifolds and their respective boundary entropies will be one of the main tasks of this paper.
%-------------------------------------------

%-------------------------------------------
\subsection{Review of Symmetric Orbifolds}\label{sec:SymNrev}
%-------------------------------------------

In this section, we will review the features of symmetric orbifolds that will be relevant for this paper. Consider a two-dimensional conformal field theory $\mathcal{C}$, with central charge $c_{\mathcal{C}}$. From this choice of \textit{seed} theory, one can construct many new CFTs. First, one can consider the $N$-fold tensor product $\mathcal{C}^{\otimes N}$, which has a large central charge $N c_{\mathcal{C}}$ when $N$ is large. From the direct product theory, one can then build new CFTs by orbifolding this theory by its global permutation symmetry (or subgroup thereofs). The symmetric orbfiolds are defined as
\be
\mathcal{C}_N \equiv \frac{\mathcal{C}^{\otimes N}}{S_N} \,.
\ee 
The orbifold procedure projects down to the states of the product theory that are invariant under the action of the symmetric group, drastically reducing the number of states. New states are also produced in the process, which are known as twisted sectors since the boundary condition of the fields get twisted by the action of the group. There is one twisted sector per conjugacy class of the symmetric group $S_N$, which are in one to one correspondence with Young diagrams.

The seed theory completely determines the orbifold theory, such that its correlation functions or partition functions are constructed from fundamental building blocks of the seed (which typically involves seed correlation functions on higher genus surfaces). Therefore, once the seed theory is known, the orbifold theory is completely specified. A simple example is the torus partition function of the orbifold theory, which is given by \cite{Ginsparg:1988ui}
\bea\label{orb}
Z_{\mathcal{C}_N}&=&\sum_{\substack{g,h\in S_N \\ gh=hg}} \begin{array}{r}\\  h~  \begin{array}{|c|}\hline ~ \\ \hline \end{array} \\  g~ \end{array} \notag \\
&=& \sum_{[g]} \frac{1}{|C_g|}\sum_{h\in C_g}\begin{array}{r}\\  h~  \begin{array}{|c|}\hline ~ \\ \hline \end{array} \\  g~ \end{array} \,.
\eea
The first sum is over all conjugacy classes and the $g$ are elements of $S_N$ drawn from the different conjugagy classes. This is the sum over twisted sectors. One also needs to sum over the centralizer of $g$ $C_g$ whose elements are labelled by $h$. The box represents a torus partition function, twisted by $h$ and $g$ in the euclidean time and spatial directions. For a given choice of $h$ and $g$, this will split into a product of torus partition functions, some of which will have twisted modular paramaters. A more detailed formula can be found in \cite{Bantay:1997ek}. We will see that the overlaps of boundary states in the orbifold theory follow a similar structure. 

The most useful information one extracts from \rref{orb} is the scaling dimension of operators in the twisted sector. The ground state of the twisted sector becomes a new primary operator, and for a single cycle of length $k$ the weight is given by
\be
h= \frac{c}{24}\left(k-\frac{1}{k}\right) \,.
\ee
Other primary operators in the twisted sector are in one to one correspondence with primary operators of the seed theory. For example, a scalar seed theory operator with weight $h_{\text{seed}}$ leads to a primary operator with weight
\be \label{htwisted}
h_{\text{twisted}}=\frac{c}{24}\left(k-\frac{1}{k}\right)+\frac{h_{\text{seed}}}{k} \,.
\ee
For more complicated twisted sectors involving multiple cycles, one simply adds the weight of each cycle.

It is also important to mention the chiral algebra of symmetric orbifolds. Assuming that the seed theory has algebra $\mathcal{A}$ (we will mostly take $\mathcal{A}=\textrm{Vir}$), the product theory simply has the symmetry $\mathcal{A}^{\otimes N}$. By performing the orbifold, the chiral algebra gets reduced to
\be \label{algebra}
\frac{\mathcal{A}^{\otimes N}}{S_N} \,.
\ee
This chiral algebra would be an important ingredient in using the tools of rational conformal field theory to find the boundary states of the orbifold theory. 
In this paper, however, we will opt to give an explicit construction that does not (directly) use these techniques. 

In the twisted sector, there are also fractional Virasoro modes. These modes are defined as \cite{Burrington:2018upk}
\be \label{fractionalL}
L_{-m/k}=\oint \frac{dz}{2\pi i} \sum_{j=1}^N T^{j}(z) e^{-2\pi i (j-1)/k} z^{1-m/k} \,.
\ee
These modes satisfy the algebra
\be
\left[L_{\frac{n}{k}},L_{\frac{n'}{k}}\right]=\frac{n-n'}{k}L_{\frac{n+n'}{k}}+\delta_{n+n',0} \frac{cN}{12} \left(\left( \frac{n}{k}\right)^2-1\right)\frac{n}{k} \,.
\ee
Some of these modes will generate new primary operators of the full chiral algebra \rref{algebra}, but only a finite number. The other operators generated from these fractional modes will be descendants.

Symmetric orbifolds obey nice properties at large $N$. In the direct product theory, the number of states at fixed energy grows with $N$. In the orbifold theory, there is a finite number of states at any fixed energy as $N\to\infty$ so the large $N$ limit converges. Moreover, the free energy of the theory is universal at large $N$ and mimicks the Hawking-Page phase transition in AdS$_3$ \cite{Keller:2011xi}.\footnote{Note that this is no longer true for higher genus partition functions \cite{Belin:2017nze}.} Correlation functions also display large $N$ factorization \cite{Pakman:2009zz}. These theories can thus serve as good proxies for true holographic CFTs (with some limitations of course).

%%%%%%%%%%%%%%%%%%%%%%%%%%%%%%%%%%%%%%%%%%%%
\section{Boundary States of Symmetric Orbifolds}\label{sec:bdrystates}
In this section, we will construct the boundary states of the symmetric orbifold theory. We will first write a general ansatz for the boundary states, and then proceed to prove that they satisfy the Cardy conditions.

\subsection{An Ansatz for the Boundary States}

We start by considering a seed theory $\mathcal{C}$ and assume that the spectrum of boundary states of the seed theory is known. The simplest case is to take the seed theory to be a Virasoro minimal model, or any other rational CFT since the spectrum of boundary states is better understood. However, we will be more general and also allow for irrational CFTs. We will show that if the seed theory boundary states are known, so are those of the orbifold theory.

Consider the set of boundary states of the seed theory
\be
\ket{a_i} \,, \qquad i=1,... \ ,n_b \,.
\ee
$n_b$ counts the number of boundary states of the seed theory, which formally can be infinite if the seed theory is irrational. Following \rref{closd}, we also have the seed theory overlaps
\be \label{seedoverlaps}
\bra{a_i}\Tilde{q}^{\frac{1}{2}(L_0+\Bar{L}_0-\frac{c}{12})}\ket{a_j}=Z_{ij}(\tau)  \,,
\ee
with $\tilde{q}=e^{-\frac{2\pi i }{\tau}}$, where the coefficients of $Z_{ij}$ in the $q$ expansion are positive integers and where the vacuum state is unique. We will see that these functions $Z_{ij}$ will be the fundamental building blocks of the final result for the symmetric orbifold overlaps.

We can start building boundary states of the product theory by tensoring such boundary states together. For example, we have
\be
\ket{a_{i_1}} \otimes \cdots \otimes \ket{a_{i_N}} \,.
\ee
If some of the $a_i$ are different, this state is not invariant under the permutation of the copies, and must therefore be symmetrized. One of the building blocks we will use to form good boundary states of the orbifold theory are thus the symmetrized version of these product states. They are of the form
\be \label{untwpre}
\ket{b}_{\textrm{untw}}=\sum_{g\in S_N} \ket{a_{i_{g(1)}}} \otimes \cdots \otimes \ket{a_{i_{g(N)}}} \,.
\ee
Note that these states are not properly normalized. We will fix the normalization later on when we assemble the various building blocks. By construction, the states \rref{untwpre} obviously satisfy the Ishibashi condition
\be
(L_n-\bar{L}_{-n})\ket{b}_{\textrm{untw}}=0 \,,
\ee
where $L_n=\sum_i L^i_n$ is the full Virasoro mode of the orbifold theory. It is worthwhile to mention that they actually satisfy the much stronger constrain
\be \label{extraconditions}
(L_n^i-\bar{L}^i_{-n})\ket{b}_{\textrm{untw}}=0  \quad \forall i \,,
\ee
which they inherit from the Ishibashi condition of each copy of the seed theory. This stronger condition should be viewed as an Ishibashi condition for the extended chiral algebra $\textrm{Vir}^{\otimes N}/S_N$. Since for rational seed theories the orbifold theory is rational with respect to this extended chiral algebra, but irrational with respect to Virasoro, it makes sense to start by considering only the boundary states that respect this extended algebra, as it is such states that we might hope to be able to simply classify.

Before moving on to the twisted sector states, we would like to mention that there are other conditions on the extended algebra that one can impose. In fact, any automorphism of this algebra should provide a set of boundary conditions.  One family of examples was considered by Recknagel \cite{Recknagel:2002qq} in product (but not orbifold) theories. They studied the boundary conditions satisfying
\be
(L_n^i-\bar{L}^{g(i)}_{-n})\ket{b}_{\textrm{untwisted}}=0  \quad \forall i \,,
\ee
for any choice of permutation $g\in S_N$. We will not attempt to extend such boundary conditions to orbifold thoeries in this paper, although we will discuss them further in the discussion section.

We now turn to the twisted sector building blocks. Since a twisted sector is fixed by its conjugacy class in $S_N$, which is described in terms of the cycle decomposition of permutations, it suffices to discuss individual cycles which can be later assembled. We start by defining a twist-k sector Ishibashi state
\be \label{Ishitwisted}
\ket{h}\rangle_{\textrm{twist-k}}=\sum_{n=0}^{\infty} \sum_{j=1/k}^{d_h(n)} \ket{h_{\textrm{twisted}}(k),n,j,\bar{n},\bar{j}} \,,
\ee
where the sum runs over the fractional Virasoro generators defined in \rref{fractionalL} and $h_{\textrm{twisted}}(k)$ is given in \rref{htwisted}. Note that this twisted Ishibashi state is now already a superposition of several Ishibashi states of the orbifold theory, since the sum over fractional descendants generates new primary states. The prescription of the state \rref{Ishitwisted} is simply to build them by superposing these states with the same coefficients.

By construction, these states satisfy the condition
\be \label{conditionstwisted}
(L_n-\bar{L}_{-n})\ket{h}\rangle_{\textrm{twist-k}}=0 \qquad \forall n \in \mathbb{Z}/k \,,
\ee
For integer $n$, this is simply the usual Ishibashi condition, but because of the non-integer cases, we find a stronger boundary condition. One should view this as the equivalent of \rref{extraconditions} in the twisted sector. Since the twisted sector is built in an $S_N$ invariant way, it no longer makes sense to talk about the individual $L^i$. However, the twisted sector inherits the fractional modes which descend from the non-$S_N$ invariant combinations of $T$ (for example $T^1-T^2$ in a $\mathbb{Z}_2$ orbifold). In principle, both \rref{extraconditions} and \rref{conditionstwisted} should be combined and organized in terms of the extended chiral algebra $\textrm{Vir}^{\otimes N}/S_N$. Unfortunately, to the best of our knowledge, a detailed understanding of this complicated algebra is still lacking so we will simply think of the boundary conditions in terms of the generalized Ishibashi conditions rather than in terms of the higher spin currents themselves. 

A boundary state in the twisted sector can then be obtained by superposing these twisted sector Ishibashi states. If we have a seed theory state given by
\be
\ket{a}=\sum_h a_h \ket{h}\rangle \,,
\ee
we obtain a boundary state
\be
\ket{a}_{\mathrm{twist-k}} = \sum_h a_h \ket{h}\rangle_{\mathrm{twist-k}} \, .
\ee
$\ket{a}_{\mathrm{twist-k}}$ lives in a twisted sector where $k$ copies of the seed theory are glued together. Before explicitly symmetrizing, we also want to specify which copies are being glued together in a twisted cycle. So, we will often write $\ket{a_{\sigma}}$ for $\sigma$ some $k$-cycle in $S_N$. The cycle indicates which, and in what order, individual copies are twisted together. It is then possible to combine these twisted boundary states as
\be
\ket{(a_{1})_{\sigma_1}} \otimes \cdots \otimes \ket{(a_{m})_{\sigma_m}} \,.
\ee
where the $\sigma_i$ are all disjoint cycles in $S_N$ and $\sum|\sigma_i| =N$. These types of states have not yet been  symmetrized or normalized. 
Note that if we have a state of the form 
\begin{equation}
    \ket{(a)_{\sigma_1}} \otimes \cdots \otimes \ket{(a)_{\sigma_m}} \,.
\end{equation}
for a collection of disjoint cycles $\sigma_i$ where all of the boundary conditions are the same in each copy, then we can unambiguously write this state in the simpler form
\begin{equation}
    \ket{(a)_{g}} \quad, \quad \quad  g = \sigma_1 \cdot \sigma_2 \cdots 
    \sigma_m \, .
\end{equation}

We are now almost ready to write down our formula for the orbifold Cardy states. 
The first piece of information we will use to specify a Cardy state will be $\vec{n}$, a $n_b$-dimensional vector. This vector means that the Cardy state will be constructed from $n_i$ boundary states of type $a_i$ from the seed theory. Because there are a total of $N$ copies, we require
\be
\sum_{i=1}^{n_b}n_i=N \,.
\ee
As we will see, there are many different Cardy states for a given choice of $\vec{n}$ and we need to specify more information. 
To each component of the vector, $n_i$, we will associate a subgroup $S_{n_i}\subset S_N$. Our boundary state is then determined by a choice of irreducible representation for each of the groups $S_{n_i}$.  In other words, it is specified by a partition of each $n_i$ (or equivalently a choice of conjugacy class from each $S_{n_i}$). We will denote each of these representations by the label $r_i$ and assemble them in a vector $\vec{r}$.

We now propose the following consistent boundary states of the orbifold theory:
\begin{claim}\label{claim:main}
A complete set of boundary states for the symmetric orbifold that are consistent with the additional symmetry constraints \rref{extraconditions} and \rref{conditionstwisted} are given by
\begin{align}
   | \vec{n},\vec{r} \rangle = \frac{1}{\sqrt{N!}}\sum_{h \in S_N}\prod_{i=1}^{n_b} \left(  \sum_{g_i \in S_{n_i}(N_i) } \frac{\chi^{r_i}(g_i)}{n_i!} |(a_i)_{h g_i h^{-1}} \rangle \right) \,.
   \label{eq:general-state}
\end{align}
for the complete set of vectors $\vec{n}$ and irreducible representations $\vec{r}$. Here $\chi^{r_i}$ is the character of the representation $r_i$. 
\end{claim}

This equation should be understood as follows. We first pick some boundary state $a_i$, of which there are $n_i$ copies. The Cardy conditions will require not just the product of untwisted states to appear, but a sum over all the possible ways of twisting these $n_i$ states together. To do so, we sum over permutations $g_i \in S_{n_i}(N_i)$ to generate the twisted states $\ket{(a_i)_{g_i}}$. By the notation $S_{n_i}(N_i)$ we mean permutations of $n_i$ elements starting at $N_i+1$ where $N_i=\sum_{j<i}n_j$. For example, for $S_3(5)$ we would generate states such as $\ket{a_{(678)}}$ in the twist-3 sector or $\ket{a_{(68)(7)}}$ in the twist-2 times untwisted sector. The second sum over $S_N$ simply makes the state invariant under permutations so that it is a well-defined orbifold state. The factor of $1/\sqrt{N!}$ is such that the state is properly normalized. 

The dynamical information (i.e. the correct way to superpose the Ishibashi states) is really encoded in the factors of $\frac{\chi^{r_i}(g)}{n_i!}$. $\chi^{r_i}(g)$ is the character of the representation $r_i$ evaluated on $g$. The factors of $n_i!$ can be seen as avoiding an overcounting, and is easiest to understand for the trivial representation $\chi^{r_i}(g)=1$. In that case, the state is built such that every single \textit{different} allowed combination of the copies into different twisted products appears exactly once, up to the overall factor of $1/\sqrt{N!}$. We give explicit examples in appendix \ref{sec:appendix-1} for $N\leq3$. We will now proceed to prove that the states $| \vec{n},\vec{r} \rangle$ satisfy Cardy's condition.

\subsection{Proving the Cardy Condition}

We are now ready to prove that the boundary states \rref{eq:general-state} are good Cardy states. We will proceed in three steps. We begin by considering a choice of $\vn$ where all seed states are the same. The first step will be to show that this homogeneous state, in the trivial representation, satisfies the Cardy conditions. The second step will be to verify that the homogeneous state with any choice of representation does as well. Finally the third step will be to check that any choice of $\vn,\vec{r}$ also works by using the previous cases as building blocks.

Recalling Section \ref{sec:BCFTrev}, to prove the Cardy conditions we need to show the overlap of any two states produces an open string partition function with integer multiplicites for states and a unique vacuum state iff the two states are identical. 
We will compute the overlap partition and verify these statements. 
All partition functions in this section will be open-string partition functions, so we will drop the notation indicating what channel we are working in. 

Let us begin by computing the open string partition function when all of the seed states are the same and the representations are trivial:
\begin{claim}\label{claim:trivial}
  Consider $\vec{n}=N \hat{i}$ and $\vec{n'}=N \hat{j}$, and the trivial representation for $r_i$ such that the character is identically equal to 1. The overlap is given by
\begin{align} \label{overlapallsame}
    \mathcal{Z}^{(N)}_{ij}(\tau)\equiv \braket{N \vec{i},\textrm{trivial}|\Tilde{q}^{\frac{1}{2}(L_0+\Bar{L}_0-\frac{c}{12})}|N \vec{j},\textrm{trivial}}=\sum_{\lbrace m_k \rbrace \in p(N)} \prod_{k=1}^N \frac{1}{(k^{m_k}m_k!)} Z_{ij}^{m_k}\left(k{\tau}\right) \,.
\end{align}
\end{claim}

\begin{proof}
By direct use of their definition, the overlap of our proposed boundary states is given by 
\begin{align}
    \frac{1}{{N!}^3} \sum_{h,h^\prime,g,g^\prime \in S_N}  \langle (a_i)_{h^\prime g^\prime h^{\prime \, -1}} |(a_j)_{h g h^{-1}} \rangle \,.
\end{align}
We can eliminate the sums over $h,h^\prime$ since the other sums cover the entire group $S_N$ (conjugation is an automorphism of the group). This gets rid of two factors of $N!$ and we are left with 
\begin{align}
   \frac{1}{N!} \sum_{g,g^\prime \in S_N} \langle (a_i)_{g^\prime} |(a_j)_{g } \rangle  \,.
\end{align}
which becomes
\begin{align}
   \frac{1}{N!} \sum_{g\in S_N} \langle (a_i)_{g^{-1}} |(a_j)_{g } \rangle  \,.
\end{align}
because the overlaps vanish unless the permutations are the inverse  of each other. The overlaps of the two states are determined by the sizes of the cycles in the unique decomposition of the permutation into disjoint cycles. 
We can indicate the sizes of cycles appearing in this decomposition in the form
\begin{eqnarray}
    g=(1)^{m_1}(2)^{m_2}....(N)^{m_N};\quad \sum_k k m_k=N \, .
\end{eqnarray}
The set of $m_k$ label a partition of $N$. The overlap is given by
\begin{equation}
    \prod_{k=1}^N Z_{ij}^{m_k}\left(k\tau\right) \,.
\end{equation}
The number of permutations in each conjugacy class specified by the partition
is 
\begin{equation}
\frac{N!}{|C_{S_N}(g)|} \,,
\end{equation}
where the order of the centralizer $C_{S_N}(g)$ is
\begin{equation}
    |C_{S_N}(g)| = \prod_{k=1}^N{k^{m_k}m_k!} \,.
\end{equation}
In total our formula thus becomes
\begin{align} \label{endproofallsame}
    \mathcal{Z}^{(N)}_{ij}(\tau)=\sum_{\lbrace m_k \rbrace \in p(N)} \prod_{k=1}^N \frac{1}{k^{m_k}m_k!} Z_{ij}^{m_k}\left(k{\tau}\right) \,,
\end{align}
as advertised.
\end{proof}

\begin{claim}
$\mathcal{Z}^{(N)}_{ij}(\tau)$ is a consistent partition function. (It has integer multiplicities and the vacuum appears once iff $i=j$.)
\end{claim}

\begin{proof}
While it may not be obvious that the coefficients in \rref{overlapallsame} are integers, the final expression can be recognized as the standard untwisted sector partition function in the open string channel for an orbifold CFT on a torus (i.e. the projection to the $S_N$ singlet states in the open string channel). The coefficients are thus integers. One way to see this is as a consequence of Polya's enumeration theorem \cite{Belin:2014fna}.

The vacuum only appears in the seed partition function ${Z}_{ij}$ when $i=j$. 
So it follows we get exactly one copy of the vacuum when the seed states are the same for the same reason we do when considering the standard orbifold on a torus. 

For the reader interested in a more explicit proof, it also will follow as a special case of our next proof for general representations. 
\end{proof}

The second step is to consider non-trivial representations for $r_i$ and $r_j$. We thus have to compute
\begin{align} 
    \mathcal{Z}_{ij}^{(N)r, r'}(\tau)\equiv \braket{N \vec{i},r|\Tilde{q}^{\frac{1}{2}(L_0+\Bar{L}_0-\frac{c}{12})}|N \vec{j},r'} \,.\label{nontrivial-overlap}
\end{align}
Using \rref{eq:general-state}, and following the proof of Claim \ref{claim:trivial}, it is easy to see that this overlap is given by
\begin{align}  \label{overlapallsamerrp}
    \mathcal{Z}_{ij}^{(N)r, r'}(\tau)= \frac{1}{N!} \sum_{g} \chi^r(g) \chi^{r'}(g) Z_{ij}^{[g]}\left({\tau}\right) \,,
\end{align}
where
\begin{equation}
     Z_{ij}^{[g]}\left({\tau}\right)= \prod_{k=1}^N Z_{ij}^{m_k}\left(k{\tau}\right) \, .
\end{equation}
Here the $m_k$ count the number of $k$ cycles in $g$ as before. The only change from the previous proof is the insertion of the characters. 

When the character was in the trivial representation, this was simply the untwisted sector partition function which manifestly gave integer coefficients. While it is perhaps less obvious that this quantity contains only integer coefficients, we will show that it is indeed the case, confirming that these states are good boundary states.

\begin{claim}\label{claim:same-state-diff-rep-good}
The partition function for general representations $r,r^\prime$  which we denote $\mathcal{Z}_{ij}^{(N)r, r'}(\tau)$ and which is given by equation \eqref{overlapallsamerrp}, is a good partition function.
\end{claim}

\begin{proof}
Let us count the multiplicity of all possible primary states in \rref{overlapallsamerrp}. States will be labelled by picking $l$ distinct states of the seed theory with multiplicity $M_l$ such that
\be
\sum_{l} M_l = N \,.
\ee
Such a state picks up a contribution from a term in the sum in \rref{overlapallsamerrp} if and only if
\be
g \in H=S_{M_1} \times ... \times S_{M_l} \,
\ee
so that we can assign each cycle in $g$ to contribute to one of the $M_l$.
Of course, there are many different ways to embed such a subgroup into $S_N$, and we want to sum over elements $g$ for any embedding. To take this into account, let us proceed in the following way. We will first start by summing over conjugacy classes, as was done in \rref{endproofallsame}. We thus have
\begin{align} 
    \mathcal{Z}_{ij}^{(N)r, r'}(\tau)=  \sum_{\lbrace m_k \rbrace \in p(N)} \chi^r(\{m_k\}) \chi^{r'}(\{m_k\}) \prod_{k=1}^N \frac{1}{k^{m_k}m_k!} Z(k\tau)^{m_k(g)} \,.
\end{align}
This is possible since the characters only depend on the conjugacy class. Now, we would like to zoom in on the particular contribution to a state specified by multiplicities $\{M_l\}$. The only non-zero terms will come from conjugacy classes that can be embedded in the group $H$. For reasons that will become clear shortly, instead of summing over conjugacy classes, we would like to sum over a particular instance of the group $H$ (say the one where the various $S_{M_l}$ are ordered such that the first $S_{M_1}$ contains the first $M_l$ elements and so on). We can do this by dividing by factors that take into account the number of elements in $H$ in a particular conjugacy class of $S_N$. This factor is given by $|\mathrm{Cl}(g)\cap H|$ for $g\in H$ where $\mathrm{Cl}(g)$ is the set of all elements in $S_N$ that are conjugate to $g$. To count the degeneracy of a state of multiplicity given by $\{M_l\}$, we thus have
\be
\rho(\{M_l\})= \sum_{g\in H}\frac{\chi^r(g) \chi^{r'}(g)}{|\mathrm{Cl}(g)\cap H|} \prod_{k=1}^N \frac{1}{k^{m_k}m_k!} N(g) \,,
\ee
where $N(g)$ is the number of times a state with multiplicity $\{M_l\}$ appears in $\prod Z(k\tau)^{m_k(g)}$. 

To find the combinatorical factor $N(g)$, we can think of the product of partition functions, $\prod_{k=1}^N Z(k\tau)^{m_k(g)}$ in terms of a particular representative element of the conjugacy class in $H$. 
For example, given the product $Z(2\tau)Z(4\tau)Z(\tau)^3$ as a class in $S_2 \times S_4\times S_3$ this could take the form
\begin{equation}
\label{eq:repcycle}
    \begin{tabular}{ccccc}
     $S_2$   & $\times$ & $S_4$ & $\times$& $S_3$ \\
     $(1)(2)$& &$(3456)$ & &$(78)(9)$\\   
    \end{tabular}
\end{equation}
We can similarly view the corresponding multiplicities of states which gave us the subgroup $H$ as a vector of the form. 
\begin{equation}
    (i_1 \, i_2 \quad j_3 \,j_4 \,j_5 \,j_6 \quad k_7 \, k_8 \, k_9 \,)
    \label{eq:ijks}
\end{equation}
The number of ways of getting this state from the partition function is then just the number of inequivalent assignments of the $i,j,k$ in \eqref{eq:ijks} to the cycle in \eqref{eq:repcycle}. For example, we get the obvious contribution by assigning
\begin{equation}
\label{eq:repcycle-assignment}
    \begin{tabular}{ccccc}
     $S_2$   & $\times$ & $S_4$ & $\times$& $S_3$ \\
     $(1)(2)$& &$(3456)$ & &$(78)(9)$\\ 
     $(i_1)(i_2)$ & & $(j_3\,j_4\,j_5\,j_6)$ & & $(k_7\,k_8)(k_9)$\\
    \end{tabular}
\end{equation}
but also
\begin{equation}
\label{eq:repcycle-assignment-2}
    \begin{tabular}{ccccc}
     $S_2$   & $\times$ & $S_4$ & $\times$& $S_3$ \\
     $(1)(2)$& &$(3456)$ & &$(78)(9)$\\ 
     $(k_8)(k_9)$ & & $(j_3\,j_4\,j_5\,j_6)$ & & $(i_1\,i_2)(k_7)$\\
    \end{tabular}
\end{equation}
where we have switched whether we get the two states $i$ from $Z(2\tau)$ or $Z(\tau)^2$. 
Note that every cycle must be assigned a single state label (here $i$, $j$, or $k$) since a $Z(k\tau)$ can only contribute $k$ times a single state. 
Because such an assignment is a mapping between the indices \eqref{eq:ijks} and \eqref{eq:repcycle}, we can see that it is itself a permutation. 
In order that a cycle not split indices, the permutation must map the representative cycle back into the subgroup $H$. 
So we get a potentially new assignment for every $g \in S_n$ such that $ghg^{-1} \in H$. 
In our second example assignment \eqref{eq:repcycle-assignment-2}, we can write the resulting permutation as 
\begin{equation}
\label{eq:repcycle-assignment-3}
    \begin{tabular}{cccccc}
     &$S_2$   & $\times$ & $S_4$ & $\times$& $S_3$ \\
     &$(1)(2)$& &$(3456)$ & &$(78)(9)$\\ 
     &$(k_8)(k_9)$ & & $(j_3\,j_4\,j_5\,j_6)$ & & $(i_1\,i_2)(k_7)$\\
     & & & & &\\
     $\Leftrightarrow\;$ &  $(1 \, 2)$ & & $(3\,4\,5\,6)$ & & $(7)(8)(9)$  \\  &$(i_1\,i_2)$ & & $(j_3\,j_4\,j_5\,j_6)$ & & $(k_7)(k_8)(k_9)$\\
    \end{tabular}
\end{equation}
where in the last two lines we have written the conjugated permutation and so brought all the $i,j,k$ back to their respective subgroups. 
On the other hand if we conjugate by a permutation in $H$ itself, it is just a rearrangement of the labelling of identical states, without assigning them to a new cycle (i.e. a new term in the partition function). 
For example, we get no new assignment from conjugation by $(12) \in S_2 \times S_4\times S_3$ which sends our first assignment in (\ref{eq:repcycle-assignment}) to 
\begin{equation}
    \begin{tabular}{ccccc}
     $S_2$   & $\times$ & $S_4$ & $\times$& $S_3$ \\
     $(1)(2)$& &$(3456)$ & &$(78)(9)$\\ 
     $(i_2)(i_1)$ & & $(j_3\,j_4\,j_5\,j_6)$ & & $(k_7\,k_8)(k_9)$\\
    \end{tabular}
\end{equation}

Note that, for every $g\in S_N$ such that $ghg^{-1} \in H$, we have $|H|$ redundant relabelings where instead we conjugate with  $h^\prime g$ for $h^\prime \in H$. 
We conclude that the number of unique assignments is in fact
\be
N(g)=\frac{|\mathrm{Cl}(g)\cap H|C_{S_N}(g)}{|H|}=\frac{|\mathrm{Cl}(g)\cap H|\prod_{k=1}^N k^{m_k}m_k!}{|H|} \,.
\ee
We thus find that in total we have
\be
\rho(\{M_l\})= \frac{1}{|H|}\sum_{g\in H}\chi^r(g) \chi^{r'}(g) \,.
\ee
To see that this is an integer, we use the decomposition of a character in terms of irreducible representations of a subgroup
\be
\chi^r(g)=\sum_{\vec{s} \ \textrm{irrep of}\  H} n_{\vec{s}}^r\chi^{\vec{s}}(g) \,,
\ee
where we have indicated the representation of the product group $H$ by a vector $\vec{s}$ of representations of each subgroup. We thus find
\bea
\rho(\{M_l\})&=& \frac{1}{|H|}\sum_{g\in H}\sum_{\vec{s},\vec{s}'}n_{\vec{s}}^r n_{\vec{s}'}^{r'}\chi^r(g) \chi^{r'}(g) \notag \\
&=&\sum_{\vec{s},\vec{s}'}n_{\vec{s}}^r n_{\vec{s}'}^{r'} \delta_{\vec{s},\vec{s}'} \notag \\
&=&\sum_s n_{\vec{s}}^r n_{\vec{s}}^{r'} \,,
\label{eq:homog-mult}
\eea
where we used the orthogonality of the irreducible characters of the subgroup $H$. Because the irreps of $H$ appear an integer number of times, this number is manifestly an integer.

Our last step is to show that the vacuum appears with multiplicity one iff the representations are the same. 
We immediately see that, for the vacuum to appear in the sub-partition functions, we must have $i=j$. 
The overall vacuum comes when we have the vacuum in every copy so that all the contributing states are the same. In other words, the multiplicity is counted by taking $H=G$. 
We then have the multiplicity of the vacuum is given by 
\begin{equation}
    \delta_{ij} \left( \frac{1}{|G|}\sum_{g\in G}\chi^r(g) \chi^{r'}(g) \right) =  \delta_{ij} \delta_{r,r'} \, ,
\end{equation}
which just follows from the orthogonality of the characters for irreps of $S_N$ (the Schur orthogonality relations). 
\end{proof}

There is a particularly nice interpretation for the integer multiplicites of states. They are given by the number of ways to split the representations $r,r'$ into irreps of subgroup specified by the particular state multiplicity $\{M_l\}$ we have considered. For $r=r'=\textrm{trivial}$, the number is always $1$, which makes complete sense, each multiplicity $\{M_l\}$ is counted exactly one time. That is simply the untwisted sector partition function. But for non-trivial representations $r,r'$ this number may be different.

 We are only left now to assemble these building blocks when we pick several different boundary states of the seed theory:

%-------------------------------------------
\begin{claim}\label{claim:general}
The overlap of two general boundary states is given by
\begin{equation} \label{ansatzbdrystate}
    Z_{\vec{n}\vec{n}^\prime}^{\vec{r}\vec{r}^{\;\prime}}(\tau) =\sum_{\lbrace c_{ij} \rbrace} \left(\sum_{s,s^\prime \ \textrm{irrep of}\  H_{c_{ij}}} n_{\vec{s}}^{\vec{r}} n_{\vec{s}^\prime}^{\vec{r}^\prime}\prod_{i,j} \mathcal{Z}^{(c_{ij}){s_i,s_j^\prime}}_{ij}(\tau) \right)\,,
\end{equation}
where $\sum_i c_{ij} = n_{j}$ and $\sum_j c_{ij}= n_i$, $H_{c_{ij}} = \prod_{i,j}S_{c_{ij}}(C_{ij})$ and $C_{ij}=\sum_{k=1}^{N}\sum_{l=1}^{j-1}c_{kl}+\sum_{k=1}^{i-1}c_{kj}$, and where $n_{\vec{s}}^{\vec{r}}$ are branching ratios for the product of irreps $\vec{r}$ into the product of irreps $\vec{s}$ of $H_{c_{ij}}$. 
\end{claim}

The set of $\lbrace c_{ij} \rbrace$ represents the possible choices of pairing between states. A given $c_{ij}$ indicates that $c_{ij}$ states of type $i$ on the one side are paired with states of type $j$ on the other side. The $\mathcal{Z}^{(c_{ij}){s_i,s_j^\prime}}_{ij}$ correspond to partition functions \eqref{overlapallsamerrp}, but with $N$ replaced by $c_{ij}$.

\begin{proof} 
We prove the claim by considering term by term the sum over $\lbrace c_{ij} \rbrace$ in \eqref{ansatzbdrystate} and counting the  multiplicities of contributions from the two states $\ket{\vec{n},\vec{r}}$ and $\ket{\vec{n^\prime},\vec{r^\prime}}$.  

For a fixed choice of $\lbrace c_{ij} \rbrace$, a permutation in the sum in \eqref{eq:general-state} contributes precisely when, for the bra vector, we have that each $g_i$ is embeddable in a subgroup
\begin{equation}
g_i \in \underset{{{j}}}{\prod} S_{c_{ij}} \subseteq S_{n_i} \, .
\end{equation}
(Or, if we are considering the ket vector, we likewise must have the equivalent condition holding for $g'_j$.) When this is the case we are able to split up the cycles appearing in $g_i$ into further subgroups $S_{c_{ij}}$. 

We will also decompose the character appearing in the full expression of each states in terms of irreducible representations of $H_{c_{ij}}$ using 
\be
\chi^{\vec{r}}(g)=\sum_{\vec{s} \ \textrm{irrep of}\  H_{c_{ij}}} n_{\vec{s}}^{\vec{r}}\chi^{\vec{s}}(g) \,.
\ee 
We then have that the relevant terms for our choice of $\lbrace c_{ij} \rbrace$ can be rearranged in the form 
\begin{align}
    &\ket{\vec{n},\vec{r}}_{|_{H_{c_{ij}}}}=\sum_{\vec{s}} n_{\vec{s}}^{\vec{r}}\ket{\vec{n},\vec{s}}\,,\\
    &\ket{\vec{n^\prime},\vec{r^\prime}}_{|_{H_{c_{ij}}}}=\sum_{\vec{s}^{\;\prime}} n_{\vec{s}^{\;\prime}}^{\vec{r}^{\;\prime}}\ket{\vec{n^\prime},\vec{s^\prime}}.
\end{align}
where we have used the same notation as in the definition of our original state:
\begin{align}
\label{form-1}
    | \vec{n},\vec{s} \rangle& = \frac{1}{\sqrt{N!}}\sum_{h \in S_N}\prod_{i,j} \left( \frac{1}{{c_{ij}!}} \sum_{g_{ij} \in S_{c_{ij}}(C_{ij}) }  \chi^{s_i}(g_{ij})|(a_i)_{h g_{ij} h^{-1}} \rangle \right) \,,\\
    | \vec{n},\vec{s^\prime} \rangle &= \frac{1}{\sqrt{N!}}\sum_{h \in S_N}\prod_{i,j} \left( \frac{1}{{c_{ij}!}} \sum_{g^\prime_{ij} \in S_{c_{ij}}(C_{ij}) }  \chi^{s^\prime_j}(g^\prime_{ij})|(a_i)_{h g^\prime_{ij} h^{-1}} \rangle \right) \,.\label{form-2}
\end{align}
The careful reader will notice that the factors of $n_i!$ we might have expected to appear have been replaced with a product of $c_{ij}!$. We have simply fixed a fiducial ordering of the $S_{c_{ij}}$ in $S_{n_i}$ so that they are mapped into the correct corresponding subgroup of the other state. The change in prefactor accounts for the multiple ways these subgroups could have been embedded. 

Let's now compute the overlap of the relevant piece of the bra and ket of the form \eqref{form-1} and \eqref{form-2}. It is  given by
\begin{equation}
    \frac{1}{N!}\sum_{h,h^\prime \in S_N} \prod_{k,l} \left( \frac{1}{{c_{kl}!}} \sum_{g^\prime_{kl} \in S_{c_{kl}}(C_{kl}) }  \chi^{s^\prime_j}(g^\prime_{kl})\langle b^l_{h^\prime g^\prime_{kl} h^{\prime\,-1}} | \right) \prod_{i,j} \left(\frac{1}{{c_{ij}!}}\sum_{g_{ij} \in S_{c_{ij}}(C_{ij}) }  \chi^{s_i}(g_{ij})|b^i_{h g_{ij} h^{-1}} \rangle \right) \,.
\end{equation}
We can eliminate one of the sums over $h^\prime$ immediately and drop the factor of $N!$ by just keeping one set of indices fixed. 
We also don't want to sum over the remaining $h$ that don't pair the appropriate groups of boundary conditon $i$ with $j$. 
Thus this simplifies to
\begin{equation}
   \sum_{h \in \prod_{kl}  S_{c_{kl}}} \prod_{k,l} \left(\frac{1}{{c_{kl}!}} \sum_{g^\prime_{kl} \in S_{c_{kl}}(C_{kl}) }\chi^{s^\prime_j}(g^\prime_{ij}) \langle b^l_{g^\prime_{kl}}| \right) \prod_{i,j} \left(\frac{1}{{c_{ij}!}} \sum_{g_{ij} \in S_{c_{ij}}(C_{ij}) }\chi^{s_i}(g_{ij}) |b^i_{h g_{ij} h^{-1}} \rangle \right) \,.
\end{equation}

Now, since we aren't permuting the different sectors we can rearrange the products first as 
\begin{equation}
    \prod_{i,j} \left( \frac{1}{{c_{ij}!}^2}  \sum_{h,g_{ij},g^\prime_{ij} \in S_{c_{ij}}(C_{ij}) }  \chi^{s^\prime_j}(g^\prime_{ij})\chi^{s_i}(g_{ij}) \langle b^j_{g^\prime_{ij}} |b^i_{h g_{ij} h^{-1}} \rangle \right) \,,
\end{equation}
and then as
\begin{equation}
    \prod_{i,j} \left( \frac{1}{{c_{ij}!}}  \sum_{g_{ij},g^\prime_{ij} \in S_{c_{ij}}(C_{ij}) } \chi^{s^\prime_j}(g_{ij})\chi^{s_i}(g_{ij}) \langle b^j_{g^\prime_{ij}} |b^i_{ g_{ij} } \rangle \right) \,.
\end{equation}
We immediately recognize this as the overlap that contributes to \eqref{overlapallsamerrp} with $(N)j,j^\prime$ replaced by $(c_{ij})s_i,s^\prime_j$
\begin{equation}
    \prod_{i,j} \mathcal{Z}^{(c_{ij})s_i,s^\prime_j}_{ij}(\tau) \,.
\end{equation}
Gathering everything together we finally obtain,
\begin{equation}
    \label{ansatzbdrystate-2}
    Z_{\vec{n}\vec{n}^\prime}^{\vec{r}\vec{r}^{\;\prime}}(\tau)\equiv \bra{\vec{n},\vec{r}}|\Tilde{q}^{\frac{1}{2}(L_0+\Bar{L}_0-\frac{c}{12})}\ket{\vec{n^\prime},\vec{r^\prime}} =\sum_{\lbrace c_{ij} \rbrace} \left( \sum_{\vec{s},\vec{s}^\prime \ \textrm{irrep of}\  H_{c_{ij}}} n_{\vec{s}}^{\vec{r}} n_{\vec{s}^\prime}^{\vec{r}^\prime}\prod_{i,j} \mathcal{Z}^{(c_{ij}){s_i,s_j^\prime}}_{ij}(\tau)\right)\,.
\end{equation}
\end{proof}

 Our final step is now to confirm that these most general states generate good partition functions:

 \begin{claim}
 The overlaps $Z_{\vec{n}\vec{n}^\prime}^{\vec{r}\vec{r}^{\;\prime}}(\tau)$ defined in Claim \ref{claim:general} are good partition functions.
 \end{claim}

 \begin{proof}
 Because each of the $\mathcal{Z}^{(c_{ij}){s_i,s_j^\prime}}_{ij}$ has integer coefficients, so do the products and sums of them that appear in \rref{ansatzbdrystate}. 
 Thus we will have integer multiplicities for all states.

It only remains to show the vacuum appears precisely once iff the states are the same. 
As was shown in the proof of Claim \ref{claim:same-state-diff-rep-good}, we only get the vacuum appearing in the sub-partition functions in \eqref{ansatzbdrystate} if $i=j$. For there to exist a $c_{ij} \propto \delta_{ij}$ so that this holds true for every term in the product we must have that $\vec{n} =\vec{n}'$. 
Moreover, when we look at this equal pairing, we do not have to expand our representations into representations of smaller subgroups $S_{c_{ij}}$. The pairing that could potentially generate a vacuum contribution then has the form
\begin{equation}
    \prod_{i} \mathcal{Z}^{(n_i)r_i,r_i^\prime}_{ii}(\tau)
\end{equation}
We also had from before that the vacuum will appear in the sub-partition functions iff the are in the same representation, so that we must have $r_i = r'_i$ for all $i$. 
Thus we conclude that $\vec{n} = \vec{n}'$ and $\vec{r} = \vec{r}'$. 
Lastly, since there is a single pairing that matches the same states together on both sides, we get precisely one copy of the vacuum, as desired. 
\end{proof}

This concludes the proof that we have identified a consistent set of Cardy boundary states of the symmetric orbifold. Finally, note that for every set of $n_i$ copies of the same boundary state $a_i$ we have an Ishibashi state in the orbifold theory for every conjugacy class that links them together into a product of twisted cycles. 
But we have also now shown that for each $n_i$ copies of state $a_i$ we get one boundary state per choice of irreducible representation. 
Thankfully, these irreps are also counted by the same conjugacy classes. 
We are able to conclude we have found as many Cardy states as Ishibashi states and thus our set of Cardy states is \textit{complete} (in this symmetry class).

We now turn to the discussion of their properties following the dictionary of AdS/BCFT. 

%%%%%%%%%%%%%%%%%%%%%%%%%%%%%%%%%%%%%%%%%%%%

%%%%%%%%%%%%%%%%%%%%%%%%%%%%%%%%%%%%%%%%%%%%
\section{AdS/BCFT Data}\label{sec:branetension}

In this section, we will discuss the properties of the boundary states we have just constructed. We will be interested in quantities that are relevant for the AdS/BCFT dictionary. The first quantity will be the boundary entropy, which is mapped to the tension of the ETW brane in the bulk, assuming that is a good description of the BCFT dual. More generally, the boundary entropy should give a rough measure of the depth of the bulk geometry. The second quantity will be the one-point function of single-trace operators. This quantity probes whether bulk matter is turned on in the geometry or not. 

Our goal will be to determine whether a typical boundary state of the orbifold theory has the properties expected of a well-behaved bulk geometry. In particular, we wish to see if typical boundary states are consistent with a local ETW brane.  As we will see, it will be important to distinguish between seed theories with a finite or an infinite number of boundary states, as the resulting outcome will be significantly different. We start by considering the finite seed CFTs.

\subsection{Finite Seed CFTs}

We expect to find a finite number of seed theory boundary states when there are a finite number of Virasoro primary operators. We also expect a finite number of seed states when we consider any rational CFT, but restrict ourselves to those states that respect a particular fixed automorphism of the extended chiral algebra. 

With a finite number of seed boundary states to choose from, in the large $N$ limit the choice of the vector $\vn$ will involve mostly repeated seed theory boundary states. Up to $1/N$ corrections, we can therefore take all the seed boundary states to be the same (i.e. $\vn=N \hat{i}$) and analyze the AdS/BCFT data for such states.\footnote{It's only a small amount of additional work to analyze the truly typical case where each $n_i \approx N/n_b$. The reader can quickly check for themselves that the answer has the same leading large-$N$ behaviour as we compute for the homogeneous case.} The notion of typicality will now constrain what type of representation vector $\vec{r}$ we pick. Since we pick all states to be the same, it involves understanding what a typical representation of $S_N$ looks like. Fortunately, much is known on the topic.

\subsubsection{The Boundary Entropy}

The first quantity we will be interested in computing is the boundary entropy, defined as
\be
g_b=\log\braket{b|0} \,.
\ee
For this boundary entropy to measure the depth of a corresponding well-behaved bulk geometry, we require
\be
g = B N \,,
\ee
where $|B| \leq  \mathcal{O}(1)$. Its sign determines whether the bulk geometry covers more $(B>0)$ or less $(B<0)$ than half of the AdS geometry. 

We would now like to determine the boundary entropy for the states \rref{eq:general-state}, for $\vn=N \hat{i}$ and for some typical representation $r$. There are numerous known results for the distribution of representations of $S_N$ in the large $N$ limit. We will review the relevant results and more details can be found in \cite{kerov2003asymptotic}.

To find the boundary entropy, note that the closed string channel vacuum only has overlap with the product state, whose conjugacy class is just the identity permutation. We would thus like to know the value of $\chi^h(1)$. Note that
\begin{equation}
    \chi^h(1) = d_r \,,
\end{equation}
which is simply the dimension of the representation. The boundary entropy reads
\begin{equation} \label{boundaryentrational}
   g^r_i =  \log \braket{0 | N \hat{i},r \hat{i}} = -\frac{1}{2} \log (N!) + \log(d_r) + N g_i \,,
\end{equation}
where $g_i$ is the entropy of the seed boundary state.
One can immediately derive a bound on $d_r$ by noting that
\begin{eqnarray} \label{sumreps}
        \sum_{\textrm{reps} \ r} d_{r}^2 = N! \,,
\end{eqnarray}
so that we have
\begin{equation}
    g^r_i < N g_i \, . 
\end{equation}

An even more useful consequence of \rref{sumreps} is that 
\begin{eqnarray} \label{plancherel}
        \sum_{\Lambda} \frac{d_{\Lambda}^2}{N!} = 1 \,,
\end{eqnarray}
so that there is a natural probability measure $\mu_N$ on the set of representations of $S_N$ called the Plancherel Measure where the probability of a representation is given by\footnote{It is interesting to note that the question of typical states labelled by a Young diagram has also appeared in other contexts, for example in the case of LLM geometries \cite{Balasubramanian:2005mg}. In that context, the authors use a different measure on the space of representations (the flat measure). It would be interesting to understand the difference between typical Young diagrams with respect to the two measures.  }
\begin{equation}
p(r) = \frac{d_{r}^2}{N!} \, .
\end{equation}
In particular, note that we can then rewrite
\begin{equation} \label{boundaryentrational-2}
   g^r_i = \frac{1}{2} \log(p(r)) + N g_i \,,
\end{equation}
so that the boundary entropy is just the log probability of the representation.  Moreover, the most probable representation will be the one that maximizes the boundary entropy. 
It's also useful to note that the average boundary entropy  is then given by the entropy of the Plancherel measure itself,
\begin{equation}
\langle g_i^r \rangle = -\frac{1}{2}S(p(r)) + Ng_i \, ,
\end{equation}
where $S$ is the usual Shannon entropy of the distribution. 

There are many nice results about typical properties of representations of the symmetric group with respect to the Plancherel measure \cite{kerov2003asymptotic}. For our purposes, the following theorem from \cite{kerov2003asymptotic} is particularly useful:
\begin{thm}
There exist positive constants $c_0$ and  $c_1$ such that 
\begin{equation}
    \lim_{N\rightarrow \infty} \mu_N \lbrace r : c_0 < -\frac{2}{\sqrt{N}} \log\left( \frac{d_r}{\sqrt{N!}} \right) < c_1\rbrace =1 \, .
\end{equation}
\end{thm}

From this theorem we immediately conclude that, asymptotically, almost every state has entropy
\begin{equation}
  -\frac{c_1}{2}\sqrt{N}  < g^r_i -N g_i < -\frac{c_0}{2}\sqrt{N}  \, . 
\end{equation}
It is known that these coefficients are themselves bounded $c_0 > 0.2313$ and $c_1 < 2.5651$. 
However, there is a numerically well-supported conjecture that asymptotically these bounds become tight. In other words, the conjecture holds that the Plancherel measure concentrates around the mean dimension. 
Thus we can argue that almost every boundary state will have entropy 
\begin{equation}
  g^r_i -N g_i \approx -\frac{C}{2}\sqrt{N}  \, . 
\end{equation}
for some unknown constant $0.23<C<2.56$.\footnote{Numerical evidence suggests that $C>1.8$.}

We have thus argued that for almost every choice of representation $r$ the boundary entropy is Planckian.  The factor of $-\frac{1}{2}N \log N$ in \rref{boundaryentrational} has been cancelled by the dimensionality of the representation, leaving only a subleading $\sqrt{N}$ correction to the Planckian brane tension.
To generalize this case to the truly typical case where we have chosen $\vec{n}$ at random as well as the representation, we need only replace the fixed seed entropy $g_i$ with the average of the seed entropies. 
Thus the typical boundary entropy will be 
\begin{equation}
      g^r_i \approx N\overline{g}  
\end{equation}
for $\overline{g} = n_b^{-1} \sum_i g_i$, up to sub-leading corrections in $N$.

We can also ask about extremal values of the entropy. In this case, we can draw upon another useful theorem from \cite{kerov2003asymptotic}:

\begin{thm}
There exist positive constants $c_0'$ and  $c_1'$ such that for all $N$
\begin{equation}
    e^{-\frac{c_1'}{2}\sqrt{N}}\sqrt{N!} \leq \max_{r} d_r \leq e^{-\frac{c_0'}{2}\sqrt{N}}\sqrt{N!} \
\end{equation}
\end{thm}
This shows that the maximum possible entropy isn't far from our mean value, 
\begin{equation}
    g^{r_{max}}_i < N g_i   -\frac{c_0'}{2}\sqrt{N}  \, , 
\end{equation}
so that we can't push the leading order answer in the positive direction. 
On the other hand, the dimension of the trivial representation is $1$. In this case, 
\begin{equation}
    g^{r_{min}}_i = N g_i   -\frac{N}{2} \log N
\end{equation}
and we see that our boundary has a super-Planckian negative entropy. 
In our holographic picture, this would correspond to the boundary condition eating up all but a vanishing sliver of the bulk geometry.

\subsubsection{One-point Functions}

Another quantity of interest is the one-point function of single-trace operators in the presence of the boundary. These one-point functions should diagnose whether or not bulk matter fields are turned on in the geometry. Following \rref{superpositionIshi}, the one-point functions are defined as
\be
p^{O}_i= \frac{\braket{0|O|a_i}}{\braket{0|a_i}} \,.
\ee
The seed theory one-point functions are
\be
p^{h}_i= \frac{(a_i)_h}{e^{g_i}} \,.
\ee
We will now calculate the one-point function of single-trace operators, both in the twisted and untwisted sectors.

\subsubsection*{Untwisted sector operators}
Let us define an untwisted sector primary operator as
\be \label{stuntw}
O=\frac{1}{\sqrt{N}}\sum_{I} O_h^I \,,
\ee
where $O_h$ is a seed theory primary operator who has unit two-point function. The operator \rref{stuntw} is obviously $S_N$ symmetric and is normalized such that it also has unit two-point function. It is a single-trace operator since it is built from a unique seed theory operator. A double-trace operator would involve the product of two seed theory operators, appropriately symmetrized (see \cite{Belin:2015hwa}). We now wish to compute $p^{O}_i$ for $\vn=N\hat{i}$ and $\vec{r}=r\hat{i}$. We have
\be
p^{O}_{N\hat{i},r\hat{i}} =\frac{ \bra{0}^{\otimes N} \frac{1}{\sqrt{N}}\sum_{I} O_h^I \ket{N\hat{i},r\hat{i}} }{\braket{0|N\hat{i},r\hat{i}}}=\sqrt{N} p_i^h \,,
\ee

There are two interesting observations to be made with this formula. First, we note that the one-point functions are in some sense much more universal than the boundary entropies. The final formula is completely independent of the choice of representation $r$. The second observation is the $N$-scaling of the expression. The $\sqrt{N}$ behaviour is exactly what is expected of matter fields that are turned on and that would produce $\mathcal{O}(1)$ backreaction to the geometry, meaning that the dual geometry would have $O(1)$ differences from a simple ETW brane in AdS. In particular, the geometry would not be locally AdS, even away from the brane. We now turn to twisted sector one-point functions.

\subsubsection*{Twisted sector one-point functions}

We will now discuss the one-point function of twisted sector operators. In the twisted sectors, a single-trace operator consists of a conjugacy class with a single cycle \cite{Belin:2015hwa}. The unit-normalized twisted sector ground state operator is given by (see for example \cite{Pakman:2009zz})
\be
\sigma_L = \frac{1}{\sqrt{L \ N! (N-L)!}} \sum_{g\in S_N} \sigma_{g (1 \ ... \ L)g^{-1}} \,.
\ee
We can now compute the one-point function. We find
\bea
p^{L}_{N \hat{i},r\hat{i}}&=&\frac{ \bra{0}^{\otimes N}\frac{1}{\sqrt{L \ N! (N-L)!}} \sum_{g\in S_N} \sigma_{g (1 \ ... \ L)g^{-1}} \ket{N\hat{i},r\hat{i}} }{\braket{0|N\hat{i},r\hat{i}}} \notag \\ \notag
&=&\frac{ \bra{0}^{\otimes N}\frac{1}{\sqrt{L \ N! (N-L)!}} \sum_{g\in S_N} \sigma_{g (1 \ ... \ L)g^{-1}} \frac{1}{\sqrt{N!}}\sum_{h \in S_N} \chi^{r}(h)   |(a_i)_{h} \rangle 
 }{\chi^r(1)\frac{1}{\sqrt{N!}}e^{Ng_i}} \,.
\eea
The sum over $g$ will only be non-zero if
\be
g (1 \ ... \ L)g^{-1}=h^{-1} \,,
\ee
and we can thus get rid of the sum over $g$ since we are anyway summing over the full symmetric group for $h$. This will give a factor of $N!$. The residual sum over $h$ will only be non-zero if $h=(1...L)^{-1}$, which singles out a single term. In total, we thus find
\be
p^{L}_{N \hat{i},r\hat{i}}=\sqrt{\frac{N!}{L(N-L)!}} e^{-g_i (L-1)} \frac{\chi^r(\textrm{L-cycle})}{\chi^{r}(1)} \,.
\ee

To compute this quantity, we need a new piece of data that was not necessary for the boundary entropy or the one-point function of untwisted-sector operators: the value of $\chi^r$ for a representation given by a one-single of length $L$. It turns out that there are also known results for characters of typical representations at large $N$. The most relevant theorem can be found in \cite{kerov2003asymptotic} (equation (3.2.1) and the theorem that follows). In a typical representation and at large $N$, we have
\be
\frac{\chi^r(\textrm{L-cycle})}{\chi^{r}(1)}\sim \sqrt{L} N^{-L/2} \,.
\ee
We thus find that at large $N$ and in a typical representation, we have
\be
p^{L}_{N \hat{i},r\hat{i}}=e^{-g_i (L-1)}  \,.
\ee

This result is quite remarkable. First of all, the $N$-supression of the characters in a typical representation exactly cancels the $N$-enhancement coming from the combinatorics of our state. Without this suppression from the typicality of representations, one would be left with one-point functions that grow as $N^{L/2}$ which means a backreaction that is stronger than $\mathcal{O}(1)$ in the bulk, and is actually enhanced by powers of $1/G_N$. It would be very challenging to make sense of such a geometry. Moreover, the fact that the one-point functions are order one numbers means that unlike untwisted sector operators, these VEVs will only induce a quantum backreaction to the bulk spacetime.

We now turn to seed CFTs with an infinite number of boundary states.

\subsection{Infinite Seed CFTs}

For seed CFTs with an infinite number of boundar states, the situation is completely different. Now a typical state will involve $N$ different seed theory states for $\vn$ is such that $n_i=1$ for all non-trivial entries, and thus such states contain no twisted sector states in the closed string channel. Moreover, this fixes all the representations in $\vec{r}$ to be trivial. As we will see, this drastically changes the combinatorics.

\subsubsection{The Boundary Entropy}

Let us now revisit the boundary entropy. From the expression for our boundary states \rref{eq:general-state}, it is straight forward to calculate the boundary entropy. We find
\be \label{bdryentropygeneral}
g_{\vec{n}}= \log\left[\bra{0}^{\otimes N} \frac{1}{\sqrt{N!}}\sum_{h \in S_N}\prod_{i=1}^{N}   |(a_{h(i)}) \rangle \right]=\log\sqrt{N!}+\sum_{i}n_i g_i \,,
\ee
where $g_i$ are the seed theory boundary entropies. Expanding in the large $N$ limit, we find

\be
g_{\vn}=\log\sqrt{N}+\sum_{i|n_i\neq0} g_i\sim \frac{N}{2} \log N+N\left(\bar{g_i}-\frac{1}{2}\right) \,,
\ee
where $\bar{g_i}$ is the averaged boundary entropy among the $N$ chosen boundary states. This time, we find a boundary entropy that scales as $N\log N$ as the combinatorical term dominates over the averaged seed theory boundary entropy. This means that the tension of the dual brane is super-Planckian! This large of a tension seems incompatible with a simple bulk description where we expect the euclidean distance to the end of the geometry to only be $O(N)$. 

Note that in this infinite case we can still generate an analogous spectrum to the finite case by choosing atypical boundaries built out of a small subset of the seed conditions. 
Thus, we can again get atypical boundary entropies with better behaved entropies
\begin{equation}
    g^{r}_i \approx N g_i     \, , 
\end{equation}
and again we can find solutions that seemingly eat up most of the bulk geometry with a super-Planckian negative entropy
\begin{equation}
    g^{r_{min}}_i = N g_i   -\frac{N}{2} \log N \, .
\end{equation}

\subsubsection{One-point Functions}

Let us now turn our attention to one-point functions. The first stricking difference with infinite seed theories is that the typical states have no support on twisted sectors, hence all the one-point function of twisted sector operators are zero! The untwisted sector one-point function of single-trace operators can easily be computed, and we find

\be
p^{O}_{\vn} =\frac{ \bra{0}^{\otimes N} \frac{1}{\sqrt{N}}\sum_{I} O_h^I \ket{\vn} }{\braket{0|\vn}}=\sqrt{N} \bar{p_i}^h \,,
\ee
where
\be
\bar{p_i}^h=\frac{1}{N}\sum_i p_i^h \,,
\ee
the averaged seed theory one-point function. This means that the one-point function have exactly the right order of magnitude to correspond to $\mathcal{O}(1)$ backreaction due to the matter fields in the bulk. This is the same behaviour as for rational seed theories, indicating that it is much more universal than the bounary entropy. Of course, the fact that the boundary entropy scales super-linearly with $N$ is still problematic for a geometric interpretation.

%%%%%%%%%%%%%%%%%%%%%%%%%%%%%%%%%%%%%%%%%%%%
\section{Discussion}\label{sec:discussion}
%%%%%%%%%%%%%%%%%%%%%%%%%%%%%%%%%%%%%%%%%%%%

In this paper, we have explicitly constructed a complete set of  boundary states for symmetric orbifolds, which respect the full chiral algebra of the orbifold theories. The states are specified by choosing $N$ boundary states of the seed CFT (with repetitions allowed) as well as specifying representations for the multiplicities of each seed boundary state. We proved that the states satisfy the Cardy conditions making them good boundary states of the orbifold theory. We also studied their properties at large $N$, focusing on the boundary entropies and single-trace one-point functions. We found a large difference in the boundary entropy between typical boundary states when the seed CFTs have a finite or infinite number of boundary states, while the one-point functions were similar for both cases and matched expectations from holography. We conclude with some open questions.

\subsection*{Symmetric product orbifolds and RCFT}

The construction of boundary states we presented in this paper has the advantage that it is quite explicit and requires little background to understand. 
On the other hand, symmetric product orbifolds are rational conformal field theories, for which a great deal is known in general about the spectrum of boundary states. (For a useful review, see \cite{Zuber:2000qs}.) 
It would be nice to re-derive our results using these techniques.

\subsection*{Permutation branes}

An important question to ask is whether we have truly characterized all boundary states given our class of boundary conditions. For Cardy states that respect the extended chiral algebra with a trivial automorphism identification, we believe this to be the case. We have a one-to-one map between Ishibashi states satisfying this boundary condition and boundary states.

However, we do expect there to be other boundary states specified by different choices of automorphisms $\Omega$. For $N$-fold tensor product theories without orbifolding, such states were constructed when the seed theory is a Virasoro minimal model in \cite{Recknagel:2002qq}. It is worthwhile to note that the proof that these states satisfy the Cardy conditions uses details of minimal model construction. It is therefore not clear how one would find such states for the $N$-fold product of irrational seed CFTs. 

Our construction, on the other hand, works no matter what choice of seed theory we pick, and is therefore robust (this is the typical expected behaviour of an orbifold theory - once the seed theory is known, we should have a complete knowledge of the orbifold theory). It is possible that Recknagel-type states might be found with reasonable effort for the symmetric orbifold of minimal models, but not for generic choices of seed theory. It would be interesting to understand this better.

Another important point to keep in mind is that the overlaps we have computed \rref{ansatzbdrystate} are strictly built from the untwisted sector of the open string channel, where some representations of $S_N$ are counted a certain number of times. But the spectrum in the open string channel is not altered from the product theory. It is natural to expect that there could also be twisted sector states in the open string channel, which have not appeared in our construction. Following an analogy with the states of \cite{Recknagel:2002qq}, we expect to find such twisted sector states in overlaps of the form
\be
\braket{b,\Omega|b',\Omega'} \,,
\ee
for $\Omega\neq\Omega'$. We leave the investigation of non-trivial $\Omega$ boundary states along with their properties for future work.

\subsection*{Rational vs Irrational seed CFTs}

We have seen that there is a striking difference between choosing finite and infinite seed CFTs in terms of the typicality of boundary states. For finite CFTs, the boundary entropy scales linearly in $N$ and the boundary states will have support on most twisted sectors. On the contrary, for infinite CFTs a typical boundary state will have a boundary entropy scaling like $N\log N$ and no support on the twisted sectors.

At first sight, this is quite confusing. One would think that the orbifold of an irrational CFT is closer to being holographic than the orbifold of a rational CFT, or at least, should not be worse. Indeed, picking an irrational CFT will retain some form of chaos, even if it is highly diluted at large $N$. There is no reason to expect that this makes a typical boundary state less holographic. However, our findings seem to indicate the opposite, since a boundary entropy scaling super-linearly in $N$ is puzzling from a bulk standpoint. It would be very interesting to understand this fact better.

\subsection*{Conformal perturbation theory}

Another avenue to explore is to turn on the marginal deformation that makes the $N$ copies of the orbifold theory interact. This can be done for specific seed theories with supersymmetry, like the non-linear sigma model on $T^4$ or K3, or the family of $\mathcal{N}=(2,2)$ theories described in \cite{Belin:2020nmp}. One could then study what happens to the boundary states we have classified under this deformation.

It is worthwhile to note that even at the orbifold point, we expect an infinite number of Virasoro Cardy states. We have only classified a finite number of them, focusing on those that respect the extended chiral algebra \rref{chiralalgebracond}. Away from the orbifold point, the extended chiral algebra no longer exists and one would naively expect that the boundary states we have found mix with those that do not respect the extended chiral algebra in some complicated way. It would be interesting to study this question in more detail, specifically in light of the recent technical progress on conformal perturbation theory (see for example \cite{Benjamin:2021zkn}).

\subsection*{Is there a bulk geometry?}

One of the most important questions in this endeavour is of course to understand whether the boundary states we have constructed can truly be thought of as having a description as an ETW brane, perhaps with some matter fields turned on, or even more generally as a higher-dimensional geometry. Of course, the gravity dual to a general symmetric orbifold is not expected to be well described by Einstein gravity, and would involve a theory of gravity where the string scale is of the same order as the AdS scale. In such a setup, the notion of a bulk geometry may not be so clear.

Related to this, we have found that the VEVs of untwisted sector single-trace operators lead to an $\mathcal{O}(1)$ backreaction on the geometry. In a top-down construction where some cycle in the internal manifold caps off smoothly, we would have a similar situation with an infinite tower of single-trace charged operators (under the R symmetry) that have $\mathcal{O}(\sqrt{N})$ VEVs. There is however an important difference between a true holographic theory and a symmetric orbifold. In a holographic theory, the number of single-trace operators at a given dimension grows polynomially with dimension (this can be understood essentially from the KK reduction of the internal manifold onto AdS). In a symmetric orbifold, the growth of untwisted sector operators is superpolynomial and grows as $e^{\sqrt{\Delta}}$. Therefore there are \textit{many} more operators being turned on, which could provide challenging for a geometric intrepretation.

Moreover, finding a boundary entropy to be $\mathcal{O}(N)$ and one-point functions no greater than $\sqrt{N}$ is a necessary condition, but of course not a sufficient one. There are other constraints required by having a consistent bulk causal structure \cite{Reeves:2021sab}. Such conditions could in principle be checked explicitly in our boundary states, and we leave it for future work. It would be interesting to see if in top-down constructions where the bulk worldsheet theory dual is known \cite{Eberhardt:2018ouy}, one can interpret some of these boundary states as having a dual description in terms of geometry.

\section*{Acknowledgements}

We are happy to thank Luis Apolo, Suzanne Bintanja, Alejandra Castro, Jan de Boer, Arjun Kar, Christoph Keller, Mark van Raamsdonk, Gordon Semenoff, and Joanna Karczmarek for fruitful discussions.

%%%%%%%%%%%%%%%%%%%%%%%%%%%%%%%%%%%%%%%%%%%%

\appendix
\section{Examples}\label{sec:appendix-1}
In the section, we write down expressions of the boundary states \eqref{eq:general-state}  and the open string partition functions for $N=2,3$.

\subsection*{$N=2:$}
We can take both seed theory boundary states either to be the same or different. When they are the same, the boundary states are 
\begin{align}
    \ket{(2),1}\equiv\ket{aa^1}&=\frac{1}{\sqrt{2}}\left(\ket{a_{(1)}}\ket{a_{(2)}}+\ket{a_{(12)}}\right)\,\\
  \ket{(2),2}\equiv\ket{aa^2}&=\frac{1}{\sqrt{2}}\left(\ket{a_{(1)}}\ket{a_{(2)}}-\ket{a_{(12)}}\right)\,.
\end{align}
Note that the characters take values from the character table of $S_2$ corresponding to the two irreps of the group. When the seed boundary states are different, we have
\begin{align}
      \ket{(1,1),1}\equiv\ket{ab^1}=\frac{1}{\sqrt{2}}\left(\ket{a_{(1)}}\ket{b_{(2)}}+\ket{b_{(1)}}\ket{a_{(2)}}\right)\,.
\end{align}
The open string partition functions resulting from the overlaps of these states are
\begin{align}
    &Z_{aa^1,\alpha\alpha^1}(\tau)=\frac{1}{2}\left(Z_{a\alpha}^{2}(\tau)+Z_{a\alpha}(2\tau)\right)=Z_{aa^2,\alpha\alpha^2}(\tau)\,,\nonumber\\
     &Z_{aa^1,\alpha\alpha^2}(\tau)=\frac{1}{2}\left(Z_{a\alpha}^{2}(\tau)-Z_{a\alpha}(2\tau)\right)\,,\nonumber\\
     &Z_{aa^1,\alpha\beta^1}(\tau)=Z_{a\alpha}(\tau)Z_{a\beta}(\tau)=Z_{aa^2,\alpha\beta^1}(\tau)\,,\nonumber\\
     &Z_{ab^1,\alpha\beta^1}(\tau)=Z_{a\alpha}(\tau)Z_{b\beta}(\tau)+Z_{a\beta}(\tau)Z_{b\alpha}(\tau).
  \end{align}
  \subsection*{$N=3:$}
  In this case, we have all three seed boundary states are identical, two are identical, or all three are different. The boundary states are
  \begin{align}
      \ket{(3),1}&\equiv\ket{aaa^1}=\frac{1}{\sqrt{6}}\left(\ket{a_{(1)}}\ket{a_{(2)}}\ket{a_{(3)}}+\ket{a}_2+\ket{a}_3\right)\,,\nonumber\\
       \ket{(3),2}&\equiv\ket{aaa^2}=\frac{1}{\sqrt{6}}\left(\ket{a_{(1)}}\ket{a_{(2)}}\ket{a_{(3)}}-\ket{a}_2+\ket{a}_3\right)\,,\nonumber\\
        \ket{(3),3}&\equiv\ket{aaa^3}=\frac{1}{\sqrt{6}}\left(2\ket{a_{(1)}}\ket{a_{(2)}}\ket{a_{(3)}}-\ket{a}_3\right)\,.\label{allsame}
  \end{align}
  Here, $\ket{a}_2\equiv\left(\ket{a_{(1)}}\ket{a_{(23)}}+\ket{a_{(2)}}\ket{a_{(13)}}+\ket{a_{(3)}}\ket{a_{(12)}}\right)$ and $\ket{a}_3\equiv \left(\ket{a_{(123)}}+\ket{a_{(132)}}\right)$.
Again, the coefficients multiplying the twisted sectors belonging to different conjugacy classes take value from the character table of $S_3$. The remaining states are
\begin{align}
\ket{(2,1),(1,1)}
&=\frac{1}{\sqrt{6}}\left(\ket{a_{(1)}}\ket{a_{(2)}}\ket{b_{(3)}}+\ket{a_{(1)}}\ket{b_{(2)}}\ket{a_{(3)}}+\ket{b_{(1)}}\ket{a_{(2)}}\ket{a_{(3)}}+\ket{b}\ket{a}_2\right)\nonumber\\
&\equiv    \ket{aa^1b^1}\;\nonumber\\
\ket{(2,1),(2,1)}
&=\frac{1}{\sqrt{6}}\left(\ket{a_{(1)}}\ket{a_{(2)}}\ket{b_{(3)}}+\ket{a_{(1)}}\ket{b_{(2)}}\ket{a_{(3)}}+\ket{b_{(1)}}\ket{a_{(2)}}\ket{a_{(3)}}-\ket{b}\ket{a}_2\right)\nonumber\\
&\equiv    \ket{aa^2b^1}\,,\nonumber\\
\ket{(1,1,1),(1,1,1)}&=\frac{1}{\sqrt{6}}\left(\ket{a_{P(1)}}\ket{b_{P(2)}}\ket{c_{P(3)}}+\text{permutations }\right)\equiv \ket{abc}\label{not-allsame}
\end{align}
The open string partition functions resulting from the set \eqref{allsame} are
\begin{align}
    &Z_{aaa^1,\alpha\alpha\alpha^1}(\tau)=\frac{1}{6}Z_{a\alpha}^3(\tau)+\frac{1}{2}Z_{a\alpha}(\tau)Z_{a\alpha}(2\tau)+\frac{1}{3}Z_{a\alpha}(3\tau)=Z_{aaa^2,\alpha\alpha\alpha^2}(\tau)\,,\\
    &Z_{aaa^3,\alpha\alpha\alpha^3}(\tau)=\frac{2}{3}Z_{a\alpha}^3(\tau)+\frac{1}{3}Z_{a\alpha}(3\tau)=Z_{aaa^1,\alpha\alpha\alpha^3}(\tau)\,,\\
    &Z_{aaa^1,\alpha\alpha\alpha^2}(\tau)=\frac{1}{6}Z_{a\alpha}^3(\tau)-\frac{1}{2}Z_{a\alpha}(\tau)Z_{a\alpha}(2\tau)+\frac{1}{3}Z_{a\alpha}(3\tau)\,,\\\label{integer-test}
    &Z_{aaa^1,\alpha\alpha\alpha^3}(\tau)=\frac{1}{3}Z_{a\alpha}^3(\tau)-\frac{1}{3}Z_{a\alpha}(3\tau)=Z_{aaa^2,\alpha\alpha\alpha^3}(\tau)\,.
\end{align}
In section \ref{sec:bdrystates}, we have proved that the above partition functions each state appear in an integer  number of time. It is instructive to see how it works explicitly. For example, The states that can appear in partition function \eqref{integer-test} will have dimensions of the form
\begin{align}
    h = h_i + h_j + h_k \quad \mathrm{for} \quad h_i \neq h_j \neq h_k \\
    h = 2h_i + h_j \quad \mathrm{for} \quad h_i \neq h_j \\
    h = 3h_i
\end{align}
Let's consider these three cases in turn. 
In the first case this can only come from the first term in the partition function. 
We get multiple copies (6) from this single term by summing over which of the $Z$s in the product each state comes from. This gives $6\times1/6 = 1$ copy of the state.
In the second case we can get contributions from the first two terms in the partition function, but not the third. From the first term we now get three copies and from the second we get one copy. Thus we have $3\times1/6 - 1/2 = 0$ copies of this state. 
In the third case we get contributions from all three terms, and precisely one from each. Thus we have $1/6-1/2 + 1/3 = 0$ copies of this state.\\
The partition functions involving the boundary states in \eqref{not-allsame} are
\begin{align}
    &Z_{aa^1b,\alpha\alpha^1\beta}(\tau)=Z_{a\alpha}(\tau)Z_{b\alpha}(\tau)Z_{a\beta}(\tau)+\frac{1}{2}\left(Z^{2}_{a\alpha}(\tau)+(Z_{a\alpha}(2\tau)\right)Z_{b\beta}(\tau)=Z_{aa^2b,\alpha\alpha^2\beta}(\tau)\,,\nonumber\\
    &Z_{aa^1b,\alpha\alpha^2\beta}(\tau)=Z_{a\alpha}(\tau)Z_{b\alpha}(\tau)Z_{a\beta}(\tau)+\frac{1}{2}\left(Z^{2}_{a\alpha}(\tau)-(Z_{a\alpha}(2\tau)\right)Z_{b\beta}(\tau)\,,\nonumber\\
    &Z_{abc,\alpha\beta\gamma}(\tau)=Z_{a\alpha}(\tau)Z_{b\beta}(\tau)Z_{c\gamma}(\tau)+\text{permutations $(\alpha,\beta,\gamma)$}.
\end{align}
 Finally we have
 \begin{align}
     &Z_{aa^1b,\alpha\alpha\alpha^1}(\tau)=\frac{1}{2}\left(Z^2_{a\alpha}(\tau)+Z_{a\alpha}(2\tau)\right)Z_{b\alpha}(\tau)=Z_{aa^2b,\alpha\alpha\alpha^2}(\tau)\,,\nonumber\\
      &Z_{aa^1b,\alpha\alpha\alpha^2}(\tau)=\frac{1}{2}\left(Z^2_{a\alpha}(\tau)-Z_{a\alpha}(2\tau)\right)Z_{b\alpha}(\tau)=Z_{aa^2b,\alpha\alpha\alpha^1}(\tau)\,,\nonumber\\
      &Z_{aa^1b,\alpha\alpha\alpha^3}(\tau)=Z_{a\alpha}^2(\tau)Z_{b\alpha}(\tau)= Z_{b\alpha}(\tau)\left(Z_{aa^1,\alpha\alpha^1}(\tau)+Z_{aa^1,\alpha\alpha^2}(\tau)\right)=Z_{aa^2b,\alpha\alpha\alpha^3}(\tau)\,,\nonumber\\
      &Z_{abc,\alpha\alpha\alpha^{1,2,3}}(\tau)=Z_{a\alpha}(\tau)Z_{b\alpha}(\tau)Z_{c\alpha}(\tau)\,,\nonumber\\
      &Z_{abc,\alpha\alpha^{1,2}\beta}(\tau)=Z_{a\alpha}(\tau)Z_{b\alpha}(\tau)Z_{c\beta}(\tau)+Z_{a\alpha}(\tau)Z_{b\alpha}(\tau)Z_{c\beta}(\tau)+Z_{a\beta}(\tau)Z_{b\alpha}(\tau)Z_{c\alpha}(\tau)\,.
 \end{align}

\bibliographystyle{jhep}
\bibliography{refs2}

\end{document}